\documentclass{article}

\usepackage{amsthm, amsmath, amsfonts, amssymb, epsfig, url}
\usepackage[lined, ruled, linesnumbered, noresetcount]{algorithm2e}
\DontPrintSemicolon

\newcommand{\bc}[1]{\left\{ #1 \right\}}

\newcommand{\f}[2]{\frac{#1}{#2}}

\newcommand{\ceil}[1]{\ensuremath{ \left\lceil #1 \right\rceil }}
\newcommand{\floor}[1]{\ensuremath{ \lfloor #1 \rfloor }}

\newcommand{\Z}{\ensuremath{\mathbb{Z}}}

\newcommand{\remove}[1]{}

\remove{
\newtheoremstyle{coolstyle}
    {9pt}
    {9pt}
    {\slshape}
    {}
    {\bfseries}
    {.}
    {.5em} 
    {}
\theoremstyle{coolstyle}
}

\newtheorem{theorem}{Theorem}[section]

\newtheorem{lemma}[theorem]{Lemma}
\newtheorem{corollary}[theorem]{Corollary}

\newtheorem{definition}[theorem]{Definition}
\newtheorem{specification}[theorem]{Specification}

\SetKwFunction{FAI}{FetchAndIncrement}
\SetKwFunction{FAA}{FetchAndAdd}
\SetKwFunction{CAS}{CompareAndSwap}
\SetKwFunction{SFAI}{FancyFAI}
\SetKwFunction{FAS}{FetchAndStore}
\SetKwFunction{TAS}{TestAndSet}

\begin{document}
\remove{
}

\title{A Complexity Separation Between the Cache-Coherent and Distributed Shared Memory Models}

\author{
Wojciech Golab\thanks{This research was conducted mostly during a postdoctoral fellowship
at the University of Calgary, under the supervision of Prof.\ Philipp Woelfel.
Author partially supported by the Natural Sciences and Engineering Research Council (NSERC) of Canada.} \\
Hewlett-Packard Labs\\
Palo Alto, California, USA\\
\url{wojciech.golab@hp.com}
}

\maketitle
We consider asynchronous multiprocessor systems where processes communicate
by accessing shared memory. Exchange of information among processes
in such a multiprocessor necessitates costly memory accesses called
\emph{remote memory references} (RMRs), which generate communication on the
interconnect joining processors and main memory.
In this paper we compare two popular shared memory architecture models,
namely the \emph{cache-coherent} (CC) and \emph{distributed shared memory} (DSM) models,
in terms of their power for solving synchronization problems efficiently
with respect to RMRs.  The particular problem we consider entails one
process sending a ``signal'' to a subset of other processes.
We show that a variant of this problem can be solved
very efficiently with respect to RMRs in the CC model, but not so
in the DSM model, even when we consider amortized RMR complexity.

To our knowledge, this is the first separation in terms of amortized
RMR complexity between the CC and DSM models.
It is also the first separation in terms of RMR complexity (for asynchronous systems)
that does not rely in any way on wait-freedom---the requirement
that a process makes progress in a bounded number of its own steps.


\vspace{1mm}
\noindent
{\bf Categories and Subject Descriptors:}
    B.3.2 {[Memory Structures]}: {Design Styles}, \textit{Shared memory};
    F.2.2 {[Analysis of Algorithms and Problem Complexity]}: {Nonnumerical Algorithms and Problems}.

\vspace{1mm}
\noindent
{\bf General Terms:} Algorithms, theory.

\vspace{1mm}

\newcommand{\probname}{signaling problem}
\newcommand{\ProbName}{Signaling Problem}
\newcommand{\opop}{operation}
\newcommand{\A}{\mathcal{A}}

\section{Introduction} \label{sec:intro}
Shared memory multiprocessors in the form of multi-core chips
  can be found in most servers and desktop computers today,
  as well as many embedded systems.
Due to the large gap between memory and processor speed, such systems rely heavily on
  architectural features that mitigate the relatively high cost of accessing memory.
Two models of such architectures, illustrated in Figure~\ref{fig_arch}, are
  the cache-coherent (CC) model and the distributed shared memory (DSM) model \cite{jand:surv}.
Cache-coherent systems are most common in practice, and often
  use a shared bus as the interconnect between processors and memory.
Memory references that can be resolved entirely using a processor's cache
  (e.g., in-cache reads)
  are called \emph{local} and are much faster than ones that
  traverse the interconnect (e.g., cache misses), called \emph{remote memory references} (RMRs).
The fact that any memory location can be cached by any process simplifies greatly
  the design of efficient algorithms in the CC model.
In contrast, in the DSM model memory is partitioned into modules that
  are tied to specific processors.
Different memory modules can be accessed in parallel by those processors
  using separate memory controllers, which provides superior memory
  bandwidth.
As in the CC model, we can classify memory references in the DSM model as fast local references
  versus more costly RMRs.
The classification in the DSM model is based only on the memory location
  (as opposed to the state of caches):
A reference to a memory location in a processor's own memory module is local,
  and a reference to another processor's memory module is an RMR.

    \begin{figure*}[htbp]
    \center
    \epsfig{file=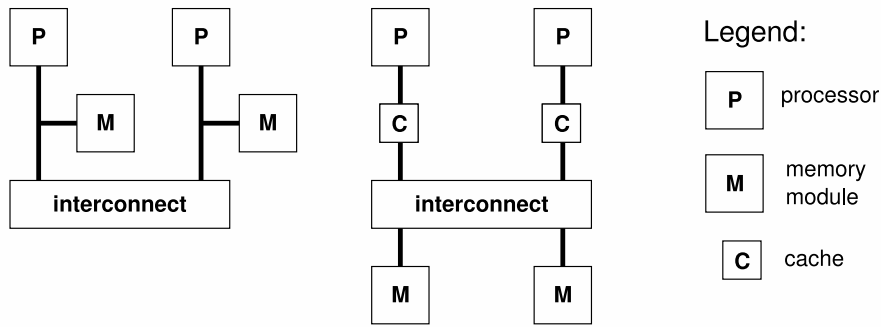, scale=0.6}
    \vspace{-12pt}
    \caption{Models of shared memory architecture---DSM (left) and CC (right).\label{fig_arch}}
    \end{figure*}

In this paper we consider efficient algorithms for solving synchronization problems
  in asynchronous multiprocessors that conform to either the CC or DSM model.
In particular, we consider algorithms that use \emph{blocking synchronization},
  whereby processes may busy-wait by repeatedly reading the values of shared variables.
In this context, RMR complexity has been shown to be a meaningful indicator of real world
  performance (e.g., \cite{tand:spin}).
The fundamental technique in the design of RMR-efficient algorithms
  is to co-locate variables with processes that access them most heavily.
Unfortunately such techniques are specific to a shared memory model.
Consequently, an algorithm that is very RMR-efficient in one model
  is not necessarily efficient with respect to RMRs in another model
  (e.g., see Section~5 of \cite{jand:surv}).

An interesting open problem is to compare the relative power of the CC and DSM
  models for solving synchronization problems efficiently with respect to RMRs.
To settle this question, we must fix a synchronization problem and a set of synchronization
  primitives (e.g., atomic reads and writes) that are available for accessing
  memory.
Consider first the mutual exclusion (ME) problem \cite{dijk:soln}, where processes
  contend for a shared resource and must coordinate with each other to ensure
  that at most one process has access to the resource at any given time.
Tight bounds for RMR complexity of $N$-process mutual exclusion are known for popular combinations
  of primitives, and do not show evidence that the CC model is more powerful
  than the DSM model, or vice-versa \cite{jand:surv,fan:lb,atti:rmr}.
That is, although the RMR complexity may depend on the combination of primitives, for each
  combination studied, the tight bound is the same for the CC model as for the DSM model.

Surprisingly, Hadzilacos and Danek \cite{danek:gme} discovered a separation
  between the CC model and DSM model by looking at the RMR complexity of solving
  \emph{group mutual exclusion} (GME).
This problem generalizes ordinary mutual exclusion by annotating each
  request for the shared resource with a \emph{session} ID,
  and allowing multiple processes to access the resource concurrently
  provided that they request the same session.
For a certain combination of primitives, it turns out that
  the RMR complexity of two-session $N$-process GME is
  less in the CC model than in the DSM model, by a factor of $\Theta(N/\log N)$.

Although we know that in one case the CC model is more powerful than the DSM
  model with respect to RMR complexity of a synchronization problem,
  the relative power of these two models is not well understood in the broader sense.
For example, we do not know whether the CC model is at least as powerful
  as the DSM model for all problems, or whether perhaps the two models are incomparable
  because for some problem DSM is more powerful than CC.
We also do not know how the two models compare under other notions
  of power, particularly the power to solve problems efficiently with
  respect to amortized (as opposed to worst-case) RMR complexity.

Answers to the above questions have interesting implications regarding
  the possibility of an \emph{RMR-preserving simulation} of
  one model using another other.
Such a simulation, if it exists, could be used to transform
  an algorithm that solves a given problem in one model to an algorithm
  that solves the same problem in another model, with at most
  a constant-factor increase in RMR complexity.
A simulation of the CC model using the DSM model would
  be particularly interesting for both software and hardware designers
  because the CC model is arguably easier to program in,
  whereas the DSM model is easier to implement in hardware.
Known results show only that such a simulation cannot exist
  if we define RMR complexity in the worst-case sense,
  leaving open the possibility that a simulation could
  at least preserve amortized RMR complexity.

Another open question is whether it is possible to show a separation
  in the RMR complexity of a problem between the CC and DSM models
  without leaning on \emph{wait-freedom}---the requirement that a process
  must make progress in a bounded number of its own steps \cite{herl:wait}.
The complexity separation shown in \cite{danek:gme} depends crucially
  on a restricted form of wait-freedom in the specification of the GME
  problem, which makes it more difficult to synchronize
  when one process releases the shared resource and allows
  a subset of other processes to emerge from busy-wait loops and make progress.
This aspect of the problem specification tends to favor the CC model, where the
  problem can be solved by having the former process signal
  the others through a single spin variable.
In contrast, in the DSM model spin variables cannot be shared
  by processes (or else RMR complexity becomes unbounded), and so
  more elaborate synchronization mechanisms must be used.
Wait-freedom restricts the possible mechanisms that can be used, and
  in that sense penalizes the DSM model.
Consequently, we wonder whether removing wait-freedom from the problem
  specification might create a more level playing field within
  which to judge the power of the CC and DSM models.

\paragraph{Summary of contributions}
The key contribution of this paper is the proof of a separation between
  the DSM and CC models in terms of the amortized RMR complexity
  of solving a simple synchronization problem.
The ``direction'' of the separation is consistent with the one
  discovered by Hadzilacos and Danek \cite{danek:gme}; the
  problem under consideration is solved more efficiently in the CC model
  than in the DSM model.
However, our result is stronger in two ways.
First, it applies to amortized RMR complexity and not only worst-case RMR complexity.
Second, it is insensitive to progress properties in the sense that it holds for both
  the wait-free version of the synchronization problem and the
  version that allows busy-waiting.

Our result implies that the CC model cannot be simulated using the DSM
  model without introducing more than a constant-factor overhead
  in terms of the total number of RMRs performed by all processes
  executing an algorithm.

\paragraph{Road map}
We give the model and definitions in Section~\ref{sec:model}.
We then survey related work in Section~\ref{sec:relwork}.
In Section~\ref{sec:pspec}, we specify a simple synchronization problem,
   called the \emph{\probname}.
In Section~\ref{sec:ubcc} we give a simple algorithm
   that solves this problem in the CC model using very few RMRs.
Section~\ref{sec:lb} presents a lower bound for the DSM model,
   which establishes a complexity separation from the CC model.
We then discuss the complexity of variations on the \probname\
   in Section~\ref{sec:ub}.
Finally, we consider the practical implications of our main result
   in Section~\ref{sec:disc}, and conclude the paper in Section~\ref{sec:conc}.

\newcommand{\Procs}{\mathcal{P}}

\section{Model} \label{sec:model}
There are $N$ asynchronous processors that communicate by accessing shared memory
   using the following atomic primitives:
   reads, writes, Compare-And-Swap (CAS) and Load-Linked/Store-Conditional (LL/SC).
(For definitions of CAS and LL/SC see \cite{jaya:llsc}.)

\paragraph{Processes and steps}
There are up to $N$ \emph{processes} running on the processors,
    at most one process per processor.
The set of processes is denoted $\Procs = \bc{p_1, p_2, ..., p_N}$,
    and we say that $p_i$ has \emph{ID} $i$.
Each process is a sequential thread of control that repeatedly applies \emph{steps},
    where each step entails a memory access and some local computation.
A step may cause a process to \emph{terminate}, meaning that it stops performing steps.
Processes can be modeled formally as input/output automata \cite{lt:ioaut},
    but here we adopt a more informal approach by describing their possible behaviors
    through a shared memory algorithm.
The algorithm is expressed through pseudo-code for each process, which is a collection
    of procedures that a process may call, one at a time.
A procedure may accept some input arguments and may return some response to the caller.
We specify a process by defining the possible sequences of procedure calls
   a process may make before terminating.
We say that a process \emph{crashes} if it terminates while performing
   a procedure call.

\paragraph{Histories}
A \emph{history} is a finite or infinite sequence of steps that describes an execution of the
    multiprocessor from well-defined initial conditions.
Process steps can be scheduled arbitrarily, and there is no bound on the
    number of steps that can be interleaved between two steps of the same process.
A process \emph{participates} in a history if it takes at least one step in that history.
A history is \emph{fair} if every process that participates either
  takes infinitely many steps, or terminates eventually.

\paragraph{Progress properties}
We will analyze the algorithms presented in this paper with respect to
   two progress properties: \emph{wait-free} and \emph{terminating}.
An algorithm is wait-free if there is an upper bound $B$ such that
   for any history $H$ of the algorithm, each (partially or fully completed)
   call to a procedure in $H$ incurs at most $B$ steps.
An algorithm is terminating if, for any fair history $H$ of the algorithm
   where no process crashes, each (partially or fully completed)
   call to a procedure in $H$ incurs a finite number of steps.
(That is, in $H$ each process that participates either terminates after completing a
   finite number of procedure calls, or else it makes infinitely many
   procedure calls.)

\paragraph{Remote Memory References}
RMRs were introduced in Section~\ref{sec:intro}.
In the DSM model, a memory access is an RMR if and only if the address accessed
  by the processor maps to a memory module tied to another processor.
In the CC model, the definition of an RMR is more complex; it depends on the
  state of each processor's cache, as well as the type of coherence protocol
  used to maintain consistency among caches.
For our purposes, we need only a loose definition of RMRs in the CC model
  that makes it possible to establish upper bounds on RMR complexity.
To that end, we assume that if a process reads some memory location
  several times, then this entire sequence of reads incurs only one RMR
  in total provided that between the first and last of these reads there is
  no nontrivial operation performed by another process on that memory location.
(A nontrivial operation overwrites a memory location, possibly with the
  same value as before.)

\section{Related Work} \label{sec:relwork}

A number of interesting complexity results have appeared in literature
  on algorithms for asynchronous shared memory multiprocessors.
Many of these pertain to mutual exclusion (ME) \cite{dijk:soln, lamp:par2},
  the problem of ensuring exclusive access to a shared resource
  among competing processes.
RMRs were originally motivated in this context as an alternative
  to traditional step complexity.
(In an asynchronous model, ME cannot be solved with bounded step complexity per process.)
The key result in RMR complexity of ME is a separation between the complexity
  (per passage through the critical section) of two classes of algorithms
  characterized by the set of primitives used.
For the class based on reads and writes,
  the tight bound is $\Theta(\log N)$ RMRs per process in the worst case
  \cite{yang:fast, kim:lock, fan:lb, atti:rmr}.
In contrast, for the class that uses reads, writes, and Fetch-And-Increment or Fetch-And-Store,
  the tight bound is $O(1)$ RMRs \cite{tand:spin, graunke:synch}.
Analogous bounds hold for first-come-first-served (FCFS) ME \cite{lamp:baker, jand:surv, dg:fcfs}.

RMR complexity bounds for the class of algorithms that use reads and writes only can be generalized
  to the class that in addition uses comparison primitives (e.g., Compare-And-Swap) \cite{jand:surv}.
For ME, the reason is that comparison primitives can be simulated efficiently using reads and writes.
For example, any comparison primitive can be implemented
  using reads and writes with only $O(1)$ RMRs per operation in the CC and DSM models
  \cite{ghw:o1l, ghhw:cas}.
Note that in such implementations \emph{every} operation incurs RMRs, in contrast
  to a comparison primitive implemented in hardware, which can sometimes be applied locally.
\emph{Locally-accessible} implementations address this issue and can be used
  to transform any algorithm that uses reads, writes, and comparison primitives
  into one that uses reads and write only, and has the same RMR complexity asymptotically \cite{ghhw:cas, golab:phd}.
Note that because this transformation necessarily introduces busy-waiting (\cite{herl:wait}),
  it can break certain correctness properties of the algorithm, such as
  bounded exit in ME, and bounded doorway in FCFS ME \cite{jand:surv}.
(This is precisely why the transformation was not used in \cite{dg:fcfs}.)

For ME and FCFS ME, the same RMR complexity bounds hold in the CC model as in the DSM model,
  with the exception of so-called Local-Failed Comparison with write-Update (LFCU) systems \cite{jand:lower}.
An LFCU system is a type of cache-coherent machine that is almost never implemented in practice.
In such systems, ME can be solved using reads, writes and Test-And-Set in
  $O(1)$ RMRs, which beats the $\Theta(\log N)$ tight bound for the DSM model.
The complexity results presented in this paper for the CC model hold just as well
  for LFCU systems as for the more standard write-through and write-back systems \cite{pathen:org}.

Mutual exclusion has been studied not only asynchronous systems, but also
  in semi-synchronous systems,
  where consecutive steps by the same process occur at most $\Delta$ time units
  apart for some $\Delta$ \cite{jand:surv}.
In one class of such systems, every process knows $\Delta$,
  and processes have the ability to delay their own execution
  by at least $\Delta$ time units in order to force others
  to make progress.
Given reads, writes and comparison primitives,
  ME can be solved in such systems using $O(1)$ RMRs in the DSM model,
  but in the CC model $\Omega(\log \log N)$ RMRs are needed in the worst case \cite{kim:timing}.
To our knowledge, this is the first result that separates the
  CC and DSM models in terms of RMR complexity for solving a fundamental synchronization problem.
(In this context we ignore complexity bounds for LFCU systems because they are not representative
  of the more common variants of the CC model.)

An interesting complexity separation has also been shown for
  the group mutual exclusion (GME) problem \cite{joung:gme} in asynchronous systems.
GME is a generalization of ME where
  requests for the shared resource are annotated with session IDs,
  and two processes can access the shared resource concurrently provided that
  they request the same session.
Several specifications for this problem have been proposed, differing in
  fairness and progress properties \cite{joung:gme, keane:gme, hadzi:gme, jpt:fgme, bhatt:gme}.
Upper bounds for RMR complexity of GME in asynchronous systems
  range from $O(\log N)$ to $O(N)$ depending on the particular specification,
  and are subject to any lower bound known for ME.
Some algorithms use only atomic reads and writes, while others rely also on CAS and/or
  Fetch-And-Add primitives.
To our knowledge, the only known lower bound on RMRs for GME, except those for ME,
  is the $\Omega(N)$ bound by Hadzilacos and Danek for the DSM model \cite{danek:gme},
  which applies to the version of GME defined by Hadzilacos \cite{hadzi:gme}
  and holds even when there are only two sessions.
This result separates the DSM model from the CC model, in which the two-session case
  can be solved using only $O(\log N)$ RMRs \cite{danek:gme}.
The ``direction'' of the separation is opposite to the one
  for ME in semi-synchronous systems.

Another line of research related to this paper pertains to transforming
  an algorithm that solves some synchronization problem in one
  shared memory model into an algorithm that solves the same problem
  in another model, with the same RMR complexity up to a constant factor.
A few transformations have been proposed for mapping mutual exclusion
  algorithms from the CC model to the DSM model \cite{hh:xform, jand:phi2}.
(In \cite{jand:phi2}, see footnote~7.)
These transformations work only for a restricted class of ME algorithms,
  and to our knowledge no general transformation exists for ME or any other
  widely studied synchronization problem.

\renewcommand{\H}{\mathcal{H}}
\SetKwData{True}{true} \SetKwData{False}{false}
\SetKwFunction{Wait}{Wait}
\SetKwFunction{Poll}{Poll}
\SetKwFunction{Signal}{Signal}

\section{Problem specification} \label{sec:pspec}

In this section we specify the \emph{\probname}, for which we
   establish RMR complexity bounds in the remainder of the paper.
The problem belongs to the family of problems where two types of processes,
   called \emph{signalers} and \emph{waiters}, must exchange information
   regarding some event (e.g., a shared resource has been released).
That is, the signalers must ensure that the waiters are aware that
   the event has occurred.
There are several important dimensions within which the safety properties for the
  \probname\ problem can be pinned down:
Is there one signaler/waiter or are there many?
Are the IDs of the signalers/waiters fixed in advance or decided arbitrarily at runtime?
How do waiters learn about the signal?

We also consider two ways to specify the semantics of the \probname.
With \emph{polling semantics}, a solution to the problem
   consists of two procedures, called \Signal{} and \Poll{}.
A signaler issues the signal by calling \Signal{}, which has no return value.
A waiter learns about the signal by calling \Poll{}, which
   returns a Boolean indicating whether the signal has been issued.
A process may call \Signal{} at most once and may call \Poll{} arbitrarily many times until
   such a call returns \True.
(Alternately, we can require that waiters and signalers be distinct.
 This has no effect on the complexity bounds presented in this paper.)
The safety properties for the procedures \Signal{} and \Poll{}
   are stated formally below in Definition~\ref{def_prob}.

\begin{specification}\label{def_prob}
For any history where each process makes zero or more calls to \Signal{} and
   zero or more calls to \Poll{}, in any order, the following hold:
\begin{itemize}
\vspace{-4pt}
\item If some call to \Poll{} returns \True, then some call to \Signal{} has already begun.
\vspace{-4pt}
\item If some call to \Poll{} returns \False, then no call to \Signal{} completed
      before this call to \Poll{} began.
\end{itemize}
\end{specification}

With \emph{blocking semantics}, a solution to the problem
   consists of procedures \Signal{} and \Wait{}.
Procedure \Signal{} is specified as for polling semantics.
A waiter learns about the signal by calling \Wait{}, which
   returns (with a trivial response) only after some call to \Signal{} has begun.
If the signal is never issued, then \Wait{} never returns.
As before, a process may call the two procedures arbitrarily many times, in any order.

To derive a complexity separation between the CC and DSM models,
   we consider one of the most difficult variations of the \probname:
   there is one signaler and there are many waiters, whose IDs are not fixed in advance;
   polling semantics are used; and waiters can terminate after
   a finite number of calls to \Poll{} even if no such call returned \True---a
   key point exploited in our lower bound proof.

Orthogonal to the above safety properties are the progress properties
   a solution may satisfy.
Terminating solutions are certainly possible, and with polling semantics
   wait-free solutions can also be considered.
Note that in terminating solutions, one process may busy-wait for another
   during a call to \Poll{} regardless of the values returned by such calls.
However, if the signal has not yet been issued, each call to \Poll{}
   must eventually terminate provided that the history is fair.

\section{Upper bound for CC model} \label{sec:ubcc}

We are interested in solutions to the \probname\
   that are efficient with respect to RMRs.
More precisely, we would like to minimize the total number of RMRs
   a process incurs across all the procedure calls it makes in a given history.
In the CC model, a very simple and RMR-efficient solution is obtained
   using a single Boolean shared variable, call it $B$, set to \False initially.
With polling semantics, procedure \Signal{} assigns $B := \True$, and \Poll{}
   reads and returns the value of $B$.
With blocking semantics, \Signal{} also assigns $B := \True$,
   and \Wait{} busy-waits until $B = \True$ holds before returning.

The solution described above is wait-free, has $O(1)$ RMR complexity per process
   in the CC model, and uses only atomic reads and writes.
As we show next in Section~\ref{sec:lb}, such a solution cannot be obtained
   in the DSM model, even if we care only about terminating solutions, settle for
   $O(1)$ amortized RMR complexity, and allow
   additional synchronization primitives.
We then discuss additional upper bounds for variations of \probname\ in Section~\ref{sec:ub}.

\section{Lower bound for DSM model}\label{sec:lb}
\newcommand{\Turan}{Tur\'{a}n}

\noindent Our main result, to which we devote the remainder of this section,
          is captured in Theorem~\ref{thm_imp} below.
(In this section, ``\probname'' refers to the particular variation
 described at the end of Section~\ref{sec:pspec}.)

\begin{definition} \label{def:ha}
For any algorithm $\A$ that solves the \probname\ with
  polling semantics, 
  let $\H_\A$ denote the set of histories (and all their finite prefixes) where
  each process makes zero or more calls to \Poll{} and zero or more calls to \Signal{},
  in arbitrary order, and then terminates.
\end{definition}

\begin{theorem} \label{thm_imp}
For any deterministic terminating algorithm $\A$ that solves the \probname\
  (with polling semantics) using atomic reads and writes,
  and for any constant $c \in \Z^+$, there exists a constant $k \in \Z^+$ and a history $H \in \H_\A$
  where $k$ processes participate and incur more than $ck$ RMRs in total in the DSM model.
\end{theorem}

At the end of this section (see Corollary~\ref{cor_imp}), we generalize the above
result to algorithms that use compare-and-swap (CAS) or load-linked/store-conditional (LL/SC)
in addition to reads and writes.

\subsection{Overview of proof}
To establish Theorem~\ref{thm_imp}, we apply a two-part proof whose
  first part is very similar to Kim and Anderson's construction
  for proving a lower bound on the RMR complexity of adaptive mutual exclusion \cite{kim:adapt}.
The key idea we borrow from them is to construct inductively
   a history where communication among processes is
   minimized using the strategies of erasing and rolling forward.
(Eventually there are no more processes to erase or roll forward,
   and the construction halts.)
The main difference between their construction and ours is the end goal.
Whereas Kim and Anderson apply the maximum possible
   number of rounds in order to trigger many RMRs,
   our goal is to apply just enough rounds so that waiters ``stabilize,''
   meaning that they stop performing RMRs and start busy-waiting on local memory.
In fact, for our purposes it helps to use as few rounds
   as possible, as that maximizes the number of waiters remaining.
Part two of our proof then shows how to extend this construction
   so that a signaler executes many (i.e., more than $ck$) RMRs
   communicating with the waiters.

In the inductive construction, we begin with all $N$ processes participating
   as waiters, making repeated calls to \Poll{}.
The strategies for erasing and rolling forward are analogous to
   Kim and Anderson's.
In our context, rolling forward means that a waiter is allowed to
   complete any ongoing call to \Poll{} it may have, and terminate.
After applying the inductive construction for only a constant number
   of rounds, there are a few processes that have been rolled
   forward, and many more ``invisible'' processes that have
   not communicated with each other.
Next, we proceed to part two.
Here we show that many of the invisible processes created in part one
   have stabilized, and that if a judiciously chosen process calls \Signal{},
   then it can be forced into an expensive (in terms of RMRs)
   ``wild goose chase.''
Since the signaler does not know who the waiters are, it must discover
   them by performing RMRs, but in that case we intervene
   and erase the waiter (if it is invisible) just before the RMR.

In the remainder of this section, we describe in more detail our two-part proof.
We proceed by contradiction, supposing that Theorem~\ref{thm_imp} is false.
That is, we suppose that some deterministic terminating algorithm $\A$
  (consisting of subroutines \Poll{} and \Signal{} for each process)
  exists that solves the \probname\ using atomic reads and writes
  with $O(1)$ amortized RMR complexity.
Then there is a constant $c \in \Z^+$ such that for any $H \in \H_\A$
  (see Definition~\ref{def:ha}),
  if $k$ processes participate in $H$ then the total number
  of RMRs incurred by these processes in the DSM model is at most $ck$.

\subsection{Proof---Part~1}\label{sec:lb:one}
We first give some definitions based on those of \cite{kim:adapt}.
These definitions are specialized for the DSM model, which makes them different from
  Kim and Anderson's more elaborate definitions that deal with the CC and DSM
  models simultaneously.

\begin{definition}
For any $H \in \H_\A$, let $Par(H)$ denote the set of processes that participate in $H$.
The set of \emph{finished processes} in $H$, denoted $Fin(H)$, is the subset
   of $Par(H)$ consisting of processes that have terminated by the end of $H$.
The set of \emph{active processes} in $H$, denoted $Act(H)$, is defined as
   $Par(H) \setminus Fin(H)$ (i.e., participating processes that have not
   yet terminated).
\end{definition}

\begin{definition}
For any $H \in \H_\A$ and any $p,q \in \!Par(H)$, we say that \emph{$p$ sees $q$ in $H$}
   if and only if $p$ reads a variable that was last written by $q$.
\end{definition}

\begin{definition}
For any $H \in \H_\A$ and any $p,q \in \!Par(H)$,
   we say that \emph{$p$ touches $q$ in $H$}
   if and only if $p$ accesses a variable local to $q$.
\end{definition}

\begin{definition} \label{def_val}
For any $H \in \H_\A$, we say that $H$ is \emph{regular}
if and only if all of the following conditions hold:
\begin{enumerate}
\vspace{-4pt}
\item For any distinct $p,q \in Par(H)$, if $p$ sees $q$ in $H$ then $q \in Fin(H)$. \label{rf1}
\vspace{-4pt}
\item For any distinct $p,q \in Par(H)$, if $p$ touches $q$ in $H$
            then $q \in Fin(H)$.  \label{rf2}
\vspace{-4pt}
\item For any variable $v$ written in $H$, if $v$ is written by more than one process
            and the last write is by $p \in Par(H)$, then $p \in Fin(H)$.  \label{rf3}
\end{enumerate}
\end{definition}

\begin{lemma} \label{lem_eras}
For any $H \in \H_\A$ and any $p \in Act(H)$, if no $q \in Par(H)$
  sees $p$ in $H$, then the history $H'$ obtained by erasing all steps
  of $p$ from $H$ is also an element of $\H_\A$.
\end{lemma}

\noindent
We will use Lemma~\ref{lem_eras} implicitly in the proof sketches that follow
   whenever we need to erase an active process from a history.

At this point we depart from Kim and Anderson's definitions \cite{kim:adapt}
   and introduce a few of our own.

\begin{definition} \label{def_stab}
For any regular $H \in \H_\A$ and any $p \in Act(H)$, $p$
   is \emph{stable} if and only if for any extension $H'$ of $H$ where
   $p$ runs solo and continues calling \Poll{} repeatedly,
   $p$ incurs zero RMRs in $H'$ after the prefix $H$.
Otherwise $p$ is \emph{unstable}.
\end{definition}

Our end goal in this part of the proof is to show that $\H_\A$ contains a regular
  history such that $Act(H)$ contains many more stable processes than $Fin(H)$.
To that end, we will construct inductively a sequence of regular computations
  $H_1$, $H_2$, $H_3$, ..., $H_c$ and prove the following invariant in the case
  when $N$ is large enough (with respect to $c$):

\begin{definition} \label{def_si}
For any $i$, $0 \leq i \leq c$, let $S(i)$ denote the statement
   that $\H_\A$ contains a regular history $H_i$ satisfying all of the following
   properties:
\begin{enumerate}
\vspace{-4pt}
\item $|Fin(H_i)| \leq i$  \label{inv:fin}
\vspace{-4pt}
\item $|Act(H_i)| \geq N^{1/3^i}$  \label{inv:inv}
\vspace{-4pt}
\item Each $p \in Act(H_i)$ incurs at most $i$ RMRs. \label{inv:rmri}
\vspace{-4pt}
\item Each unstable $p \in Act(H_i)$ incurs exactly $i$ RMRs. \label{inv:rmriu}
\vspace{-4pt}
\item Each $p \in Fin(H_i)$ incurs at most $ci$ RMRs. \label{inv:rmrf}
\end{enumerate}
\end{definition}

\begin{lemma} \label{lem_inv}
If $N$ is large enough then for any $i$, $0 \leq i \leq c$, $S(i)$ holds.
\end{lemma}

\vspace{-4pt}
\textsc{Proof sketch:}
The analysis is similar to that of \cite{kim:adapt} at a high level,
  with some technical differences.
We proceed by induction.
For $S(0)$, note that in $H_0$ there $N$ active processes and no finished processes,
   and no process has performed an RMR.
Now for any $0 \leq i < c$ suppose that $S(i)$ holds.
We must prove $S(i+1)$.
To that end, we will construct $H_{i+1}$ from $H_i$.
If there are no unstable processes in $Act(H_i)$, we simply
   let each active process take one more step,
   which is necessarily local step.

Otherwise, we let each unstable process in $Act(H_i)$ take steps until it is about to perform an RMR,
  and determine for each such process its next RMR.
We will refer to these as the ``next RMRs'' of the unstable processes.
Allowing each such process to apply its next RMR may yield a history $H_i'$ that is not
  regular because it violates one of the properties in Definition~\ref{def_val}.
As in Kim and Anderson's proof, properties~\ref{rf1} and \ref{rf2} of~Definition~\ref{def_val}
  are dealt with easily by erasing at most a constant fraction of active processes.
To that end, we construct a ``conflict graph'' where vertices represent
  active processes and an edge $\bc{p,q}$ exists if and only if
  $p$ sees or touches $q \neq p$ in its next RMR (or vice versa).
Since a process can see or touch at most one other process by performing an RMR,
  the conflict graph has average degree $d \leq 4$.
(For each process we add at most two edges, and each edge contributes
  to the degree of two vertices.)
By \Turan's theorem, the conflict graph contains an independent set
  containing at least $\f{1}{d+1} = \f{1}{5}$ of the active processes.
Keeping these and erasing the remaining active processes resolves all the conflicts.
Once this is done, we apply any next RMRs that perform a read, and consider
  the remaining RMRs for property~\ref{rf3} of~Definition~\ref{def_val}.
We consider two cases, as in Kim and Anderson's proof: one dealt
  with by rolling forward, the other by erasing.
Let $X$ denote the number of unstable processes remaining after
  properties~\ref{rf1}--\ref{rf2} of Definition~\ref{def_val} are dealt with.

\paragraph{Roll-forward case}
If at least $\floor{\sqrt{X}}$ unstable processes are about to access the same
   shared variable $v$ in their next RMRs,
   we erase all other unstable processes,
   and allow the ones remaining to apply their next RMRs on $v$
   in some arbitrary order.
The last process to write $v$, call it $r$, is then rolled forward.
As we roll forward $r$, it may see or touch other processes.
If $r$ sees or touches an active process $p$, then this is an RMR for $r$ and
   we erase $p$.
(It follows from $H_i$ being regular that $r$ cannot see $p$ via a local access
   prior to $r$'s next RMR.)
If $r$ executes more than $c(i+1)$ RMRs in total as a result of being
  rolled forward, then we obtain a contradiction easily:
Erasing all other active processes we obtain
  a history where there are at most $i+1$ finished processes,
  namely $i$ from $H_i$ (by property~\ref{inv:fin} of~Definition~\ref{def_si}) plus $r$,
  and no other processes.
The total number of RMRs in this history is more than $c(i+1)$, which contradicts
  our definition of $\A$.
Thus, by rolling forward $r$ we erase at most $c(i+1) \leq c^2 + c$ active processes.
The number of active processes remaining is at least $\floor{\sqrt{X}} - (c^2 + c + 1)$.

\paragraph{Erasing case}
If there is no single variable that is about to be accessed by at least $\floor{\sqrt{X}}$
  unstable processes in their next RMRs,
  then these RMRs collectively access at least $\floor{\sqrt{X}}$ distinct variables.
For each such variable, we choose arbitrarily an unstable process that will
  access it, and we erase all the other unstable processes.
This leaves at least $\floor{\sqrt{X}}$ active processes.
It is possible that some of the next RMRs of these remaining processes
   are about to write registers that have already been written by
   active processes.
We can eliminate this problem by erasing some of the active processes.
To that end, we construct a conflict graph where each vertex is an active
   process and an edge $\bc{p,q}$ exists if and only if
   $p$ writes in its next RMR a variable previously written by $q \neq p$.
Since each active process has at most one next RMR, the conflict
   graph has average degree $d \leq 2$.
(For each process we add at most one edge, and each edge contributes
  to the degree of two vertices.)
By \Turan's theorem, the conflict graph contains an independent set
  containing at least $\f{1}{d+1} = \f{1}{3}$ of the active processes.
Keeping these and erasing the remaining active processes resolves all the conflicts,
  leaving $\ceil{\floor{\sqrt{X}} \times \f{1}{3}}$ active processes.

\bigskip
Let $H_{i+1}$ denote the history obtained from $H_i$ by our construction.
It follows from our construction that $H_{i+1}$ is regular, and so it
   remains to show that it satisfies properties~\ref{inv:fin}--\ref{inv:rmrf}
   of $S(i+1)$ given that $S(i)$ holds.
Property~\ref{inv:fin} holds because $|Fin(H_{i+1})| \leq |Fin(H_{i})| + 1 \leq i$.
Property~\ref{inv:inv} holds because $|Act(H_{i+1})| \geq |Act(H_{i})|^{1/3} \geq N^{1/3^{(i+1)}}$ for $N$ large enough.
Properties~\ref{inv:rmri} and \ref{inv:rmriu} follow from $S(i)$ and the
   application of the next RMRs round $i+1$ of the construction.
Property~\ref{inv:rmrf} follows from $S(i)$ and the fact that any
   process being rolled forward in round $i+1$ incurs at most $c(i+1)$
   RMRs in $H_{i+1}$ under our original supposition that Theorem~\ref{thm_imp} is false.
Thus, $S(i+1)$ holds.
\qed

\subsection{Proof---Part~2}

\begin{lemma} \label{lem_halfstab}
In the history $H_c$ referred to by Lemma~\ref{lem_inv}
   there are at least $\floor{N^{1/3^c}/2}$ stable processes.
\end{lemma}

\vspace{-4pt}
\textsc{Proof sketch:}
Recall that, by Lemma~\ref{lem_inv}, the following hold:
   $|Fin(H_c)| \leq c$,
   $|Act(H_c)| \geq N^{1/3^c}$, and
   that each unstable process in $Act(H_c)$ has performed exactly $c$ RMRs in $H_c$.
Now suppose for contradiction that fewer than $\floor{N^{1/3^c}/2}$ of the active
   processes are stable.
Then at least $\floor{N^{1/3^c}/2}$ are unstable, and moreover each one
   has performed exactly $c$ RMRs in $H_c$.
Let $H_c'$ be the history obtained from $H_c$ by erasing all active
   processes except for exactly $\floor{N^{1/3^c}/2}$ unstable ones,
   and then letting each unstable process make one more RMR
   (in any order).
This history satisfies the following:
  $|Fin(H_c')| = |Fin(H_c)| \leq c$,
  $|Act(H_c')| = \floor{N^{1/3^c}/2}$,
  and each process in $Act(H_c')$ has incurred $c+1$ RMRs.
Consequently, $|Par(H_c')| \leq c + \floor{N^{1/3^c}/2}$, and the total number
  of RMRs incurred in $H_c'$ is at least $(c+1)\floor{N^{1/3^c}/2} = c \floor{N^{1/3^c}/2} + \floor{N^{1/3^c}/2}$,
  which is greater than $c$ times $|Par(H_c')|$ when $N$ is large enough.
This contradicts the RMR complexity of $\A$.
\qed

\begin{lemma} \label{lem_manstab}
There exists a regular history $H \in \H_\A$ in which there are exactly $\floor{N^{1/3^c}/2}$ stable active
  processes, at most $c$ finished processes, and no other processes participating.
Furthermore, in $H$, each process that participates incurs at most $c^2$ RMRs in the DSM model.
\end{lemma}

\vspace{-4pt}
\textsc{Proof sketch:}
By Lemma~\ref{lem_halfstab}, the history $H_c$ referred to by Lemma~\ref{lem_inv}
  contains at least $\floor{N^{1/3^c}/2}$ stable processes.
It also contains at most $c$ finished processes by Lemma~\ref{lem_inv}.
To construct $H$, take $H_c$ and erase all active processes except $\floor{N^{1/3^c}/2}$
   unstable ones.
The remaining processes incur the same number of RMRs in $H$ as in $H_c$,
  which is bounded in Lemma~\ref{lem_inv}:
  each active process in $H_c$ incurs at most $c$ RMRs,
  and each finished process incurs at most $c^2$ RMRs.
Thus, each process incurs at most $c^2$ RMRs, as wanted.
\qed

\vspace{6pt}
We now describe how to use the history $H$ referred to by Lemma~\ref{lem_manstab} to derive
  a contradiction.

\begin{lemma} \label{lem_toormr}
For large enough $N$, there exists a history $H' \in \H_\A$ and a constant $k \in \Z^+$, such that at most $k$ processes participate
  in $H$ and yet the total number of RMRs incurred in $H$ is more than $ck$.
\end{lemma}

\vspace{-4pt}
\textsc{Proof sketch:}
Let $H$ be the history referred to by Lemma~\ref{lem_manstab}.
In this history, each process that participates incurs at most $c^2$ RMRs,
  and at most $c + \floor{N^{1/3^c}/2}$ processes participate.
Thus, in total the participating processes write to at most
  $(c + \floor{N^{1/3^c}/2})(1 + c^2)$ distinct memory modules.
(Each process may write its own module, and at most $c^2$ remote ones.)
For $N$ large enough, this means there is some process $s$ whose memory module
  is not written in $H$.

Now construct $H'$ from $H$ as follows.
First, let each stable process run solo, one by one, completing any pending call to \Poll{} that it may have.
Since each process is stable, by Definition~\ref{def_stab}, it will not incur any RMRs doing so.
Now let $s$ run solo and make a call to \Signal{}, which must eventually terminate. 
As $s$ performs this call, each time $s$ is about to see a process $p$
   that is active in $H$, or is about to write
   memory local to $p$, erase $p$ and then allow $s$ to take its step.
Note that this step by $s$ must be an RMR because $s \neq p$, and because by our choice of $s$,
   process $p$ has never written memory local to $s$.
It follows that $s$ performs one such RMR for each stable process $p \in Act(H)$,
   otherwise after $s$ completes its call to \Signal{} there is some $p$ that
   remains stable and whose local memory is in the same state as at the end of $H$.
Consequently, if $p$ now calls \Poll{}, then its call returns the same response as $p$'s last call,
   namely \False, contradicting the specification of the \probname\ (see Definition~\ref{def_prob}).
Thus, $s$ performs at least $\floor{N^{1/3^c}/2}$ RMRs.
Finally, erase any active process that remains, which leaves only $s$ and at most $c$ finished
   processes participating in the history.
Let $H'$ denote this history.  Note that
   at most $k = c+1$ processes participate in $H'$, and yet $s$ incurs at least $\floor{N^{1/3^c}/2}$
   RMRs in $H'$ in the DSM model.
For large enough $N$, this means that the number of RMRs is more than $ck = c(c+1)$, as wanted.
\qed

\vspace{6pt}
\textsc{Proof of Theorem~\ref{thm_imp}:}
Lemma~\ref{lem_toormr} contradicts our assumption on the RMR complexity of $\A$.
Thus, $\A$ does not exist and Theorem~\ref{thm_imp} holds.
\qed

\begin{corollary} \label{cor_imp}
For any deterministic algorithm $\A$ that solves the \probname\ (with polling semantics)
  using atomic reads and writes, and either CAS or LL/SC,
  and for any constant $c \in \Z^+$, there exists a constant $k \in \Z^+$ and a history $H \in \H_\A$
  where $k$ processes participate and incur more than $ck$ RMRs in total in the DSM model.
\end{corollary}

\begin{proof}
Suppose for contradiction that Corollary~\ref{cor_imp} is false, namely
  that some deterministic terminating algorithm $\A$ exists that solves \probname\ using
  the stated set of base object types, and there is a constant $c \in \Z^+$ such that for any $H \in \H_\A$
  the total number of RMRs incurred in the DSM model in $H$ is at most $c$ times the number of processes participating.
Replacing the variables accessed via CAS or LL/SC in $\A$ with the locally-accessible $O(1)$-RMR implementations
  of these primitives presented in \cite{golab:phd,ghhw:cas},
  we obtain another terminating algorithm $\A'$ that solves the \probname\ using
  only atomic reads and writes, and there exists a constant $c' \in \Z^+$ such that for any $H \in \H_{\A'}$
  the total number of RMRs incurred in the DSM model in $H$ is at most $c'$ times the number of processes participating.
The existence of $\A'$ contradicts Theorem~\ref{thm_imp}.
\end{proof}

\section{Additional complexity bounds} \label{sec:ub}
In this section we discuss solutions to variations of the \probname\
   defined in Section~\ref{sec:pspec}.
Recall that there are two flavors of the problem:
   with polling semantics, waiters repeatedly call \Poll{} to determine
   if the signal has been issued;
   with blocking semantics, waiters instead call \Wait{}, which does not return
   until the signal has been issued.
Both flavors of the \probname\
   are solved easily in the CC model using $O(1)$ RMRs per process worst-case.
The solution is presented in Section~\ref{sec:ubcc}.

As one might suspect, in light of our lower bound result, the
   solution space for variations of the \probname\
   in the DSM model is somewhat more complex.
The solutions proposed for the CC model do not work ``out of the box''
   in the DSM model in the sense that they have unbounded RMR complexity.
Nevertheless, RMR-efficient algorithms for the DSM model
   can be devised using more elaborate synchronization
   techniques.
We explore such solutions in detail in the remainder of this section.
For each problem variation we consider, we will describe the polling
   solution, from which blocking solution can be achieved easily
   by implementing \Wait{} via repeated execution of the code for \Poll{}.
In some cases, a more efficient solution is possible for
   blocking semantics, as we indicate below.

\paragraph{Single waiter}\  \newline
If there is at most a single waiter, and its ID is not necessarily fixed in advance,
   the problem can be solved using two global variables,
   say $W$ (process ID, initially NIL) and $S$ (Boolean, initially \False),
   as well as an array $V[1..N]$ of Boolean variables (initially \False),
   where $V[i]$ is local to process $p_i$.
The first call to \Poll{} writes the waiter's ID to $W$, and then
   reads and returns the value of $S$.
On subsequent invocations by process $p_i$, \Poll{} reads and returns
   $V[i]$ instead.
\Signal{} sets $S$ to \True, and then reads $W$.
If NIL is read, \Signal{} returns immediately.
If a non-NIL value is read, it must be the waiter's ID, say $p_j$,
   in which case the signaler writes \True to $V[j]$.
This algorithm has $O(1)$ RMR complexity per process in the worst case,
  matching the upper bound for the CC model.

\paragraph{Many waiters, fixed in advance}\  \newline
In this formulation of the problem, the signaler knows in advance the IDs
   of the waiters that will participate eventually
   (i.e., if the history is extended by sufficiently many steps).
A simple solution in this case is to use an array of Boolean
   variables $V[1..N]$, initially all \False, $V[i]$ local to process $p_i$.
A call to \Poll{} by $p_i$ reads and returns $V[i]$,
   and \Signal{} sets $V[j]$ for each waiter $p_j$ whose ID is fixed in advance.
Letting $W$ denote the number of waiters, this solution has $O(W)$ RMR
   complexity per process worst-case.
However, amortized RMR complexity may be more than $O(1)$ RMRs
   if the signaler performs $W$ RMRs but only $o(W)$ waiters
   participate so far in the history.
(This is possible if the history is not yet long enough for every
   fixed waiter to begin participating.)

For terminating solutions, it is easy to reduce the amortized RMR complexity
   to $O(1)$ in all histories.
The signaler can simply wait for each waiter to participate, using another
   array of Boolean flags, before writing any element of $V$.
With blocking semantics, the worst-case RMR complexity per process can
   also be reduced to $O(1)$ using the work sharing techniques described
   in \cite{ghw:o1l}.

For wait-free solutions using reads, writes and comparison primitives,
   $O(1)$ amortized RMR complexity is impossible to achieve
   for all histories when $W$ is large enough (e.g., $W \in \Theta(N)$).
This result follows by a lower bound proof similar to the one
   given in Section~\ref{sec:lb}.
(Although the signaler knows which waiters will participate eventually,
   it does not know which waiters participate at the time when it
   calls \Signal{}, which must return in a bounded number of steps.)

For terminating solutions with polling semantics, it is also possible
  to show that in the worst case the signaler must perform $\Omega(W)$ RMRs
  if all $W$ waiters participate by the time \Signal{} is called.
This follows by a simplified version of the lower bound proof from Section~\ref{sec:lb}.
(This simplified proof is similar in spirit to the lower bound proof of
  Hadzilacos and Danek \cite{danek:gme}, but does not rely on any form of
  wait-freedom.)
The main idea is as follows:
First, $W$ waiters call \Poll{} repeatedly until they are all stable
   (i.e., accessing only local memory).
We then allow each waiter to complete any ongoing call to \Poll{} it may have.
Next, a signaler makes a call to \Signal{}, which
  must terminate in a finite number of steps because the signaler is running
  solo starting from a state where each waiter is not required to take
  any more steps (i.e., it is ``in between'' calls to \Poll{}).
Furthermore, before this call to \Signal{} terminates, the signaler must write
  remotely to the local memory of each waiter, except possibly itself,
  which incurs $\Omega(W)$ RMRs.
To see this, suppose that the signaler neglects to write $p_i$'s memory for
  some waiter $p_i$ different from itself.
In that case, after the call to \Signal{} completes, $p_i$ may make another call to
  \Poll{} that will incorrectly return the same value as $p_i$'s previous call
  (i.e., \False).

\break
\paragraph{Many waiters not fixed in advance, one signaler fixed in advance}\  \newline
The solution is similar to the case with fixed waiters, with a few additional steps.
Upon calling \Poll{} for the first time, a waiter ``registers'' with the
   signaler by setting a dedicated Boolean flag in the signaler's local memory.
The signaler, upon calling \Signal{}, checks for each $i$ whether $p_i$
   has registered, and if so, performs a remote write to $V[i]$.
In addition, we must handle correctly the race condition when waiters
   register while the signaler is calling \Signal{}.
To that end, it suffices to use an additional global variable
   analogous to $S$ from the single waiter case.
The signaler writes $S$ at the beginning of \Signal{}, and waiters check
   $S$ at the end of their first call to \Poll{} (i.e., after registering).

A terminating algorithm for this version of the problem appears in \cite{ghhw:cas}.
It uses using atomic reads and writes only and
  incurs $O(1)$ RMRs per process in the worst case. 

\paragraph{Many waiters not fixed in advance, one signaler not fixed in advance}\  \newline
With polling semantics the problem is subject to the lower bound from Section~\ref{sec:lb}.
That is, if only reads, writes and CAS or LL/SC are used, the problem can be solved
   more efficiently in the CC model than in the DSM model
   with respect to amortized RMR complexity.
It is possible to close this gap by using stronger primitives.
Recall from Section~\ref{sec:relwork} that if Fetch-And-Increment or Fetch-And-Add
  are available in addition to reads and writes, then it is possible to solve mutual
  exclusion using $O(1)$ RMRs per process, which can be used to construct
  a shared queue with the same RMR complexity.
Waiters and signalers can leverage such a queue as follows: During its first call to
  \Poll{}, a waiter adds itself to the queue, and also checks a global flag to see
  if a signaler has started a call to \Signal{}.
During subsequent calls to \Poll{}, a waiter only checks a dedicated flag in its own
  local memory.
During a call to \Signal{}, the signaler sets the global flag, then
  retrieves the set of all waiters from the shared queue, and writes the dedicated
  flag for each waiter found in the queue.
The worst-case RMR complexity per process of this solution is $O(1)$ for waiters, and $O(k)$
  for the signaler when there are $k$ waiters participating.
Thus, amortized RMR complexity is $O(1)$.

With blocking semantics, the problem can be reduced to the single-waiter case by having
   the waiters elect a leader, which learns about the signal and then
   ensures that the signal is propagated to the remaining waiters.
Using the synchronization techniques described in \cite{ghhw:cas},
   such a solution is possible with $O(1)$ RMR complexity per process worst-case,
   using only atomic reads and writes. 
(The leader election algorithm must tell each waiter the ID of the leader
   rather than merely telling each waiter whether it is the leader.)

\paragraph{Many waiters not fixed in advance, many signalers}\  \newline
One possibility is to reduce this case to ``one signaler not fixed in advance''
  by having signalers elect a leader that will signal the waiters.
Leader election can be solved in $O(1)$ RMRs per process worst-case using atomic reads and writes
  by a terminating algorithm \cite{ghw:o1l}, or in one step per process using virtually
  any read-modify-write primitive (e.g., Test-And-Set or Fetch-And-Store).

\section{RMRs vs. Observed Performance} \label{sec:disc}

In this section, we discuss the relationship between RMR complexity
and observed performance in real world multiprocessors,
and comment on the practical interpretation of our complexity separation result.

The RMR complexity measure attempts to quantify inter-process communication
by counting memory accesses that engage the interconnect joining processors
and memory (see Figure~\ref{fig_arch}).
Since the interconnect is slow relative to processor speed,
observed performance tends to degrade as communication over the interconnect increases.
This effect has been demonstrated in the context of mutual exclusion algorithms
  on the Sequent Symmetry (a CC machine) and the BBN Butterfly (a DSM machine)
  \cite{tand:spin, graunke:synch, mcs:algo, yang:fast}.
The key lesson from this body of literature is that
  algorithms with bounded RMR complexity (i.e., so-called \emph{local-spin} algorithms)
  outperform algorithms that have unbounded RMR complexity by a significant margin
  under worst-case conditions (i.e., maximum concurrency).

Although the benefits of local-spin algorithms are well understood,
  it is important to note that RMR complexity is not a very precise
  tool for predicting real world performance, especially for
  cache-coherent machines.
For example, to derive our upper bound in Section~\ref{sec:ubcc},
  we made the simplifying assumption that a cache behaves ``ideally,''
  meaning that it never drops data spuriously (see Section~\ref{sec:model}).
Since this assumption does not hold in a preemptive multitasking environment,
  especially under high load,
  theoretical RMR complexity bounds are prone to underestimate the actual number of RMRs.
An even more important concern is the imprecise relationship between
  RMRs and communication.
We focus on the latter issue in more detail in the remainder of this section.

Consider the simplified example of a cache-coherent system where processes
  communicate using read and write operations, and coherence is ensured
  by a \emph{write-through} protocol \cite{pathen:org}.
In such a system, a read operation on a shared object $O$ either finds
  a cached copy of $O$ locally, or else it fetches $O$ from main memory
  and creates a copy of $O$ in the local cache.
A write operation on $O$ applies the new value for $O$ to main memory,
  creates (or updates) a copy of $O$ in the local cache, and
  invalidates (i.e., destroys) all copies of $O$ in remote caches.
Thus, while an RMR on read generates a fixed amount of communication,
  an RMR on write may trigger multiple ``invalidation messages.''
This is in contrast to the DSM model where, in the absence
  of a coherence protocol, any RMR generates a fixed amount
  of communication.

The above example illustrates that RMRs in the CC model and
  RMRs in the DSM model are very different ``currencies''
  for describing the cost of an algorithm.
Consequently, a meaningful comparison between them
  requires that we define the ``exchange rate.''
To that end, we must fix a particular cache-coherent architecture
  and coherence protocol.
The simplest scenario occurs when the interconnect joining processors and
  memory is a shared bus, and any message generated by the coherence protocol
  is broadcast to all processors.
In that case, a single message suffices to invalidate all remote copies of an
  object on write, and so RMRs in this type of system are
  ``at par'' with RMRs in the DSM model.

Since a shared bus offers limited communication bandwidth,
  realistic large-scale cache-coherent systems use more elaborate interconnects.
In such systems, an RMR on write may generate multiple invalidation messages,
  the exact number depending on the topology of the interconnect and
  the state information maintained by the coherence protocol.
Consequently, RMR complexity per process can underestimate vastly
  the amount of communication triggered by a process.
But what about amortized RMR complexity, which is the focus of this paper?

A key observation is that in any cache-coherent system, a cached copy
  of a shared object can be invalidated at most once, because
  invalidation destroys it.
Since an RMR is necessary to create a cached copy in the first place,
  this means that the total number of \emph{invalidations}
  (i.e., events where a cached copy of an object is destroyed)
  is bounded from above by the number of RMRs.
The implications of this on message complexity depend on
  how the number of invalidations relates to the
  number of invalidation messages that trigger them.

Ideally, the cache coherence protocol would maintain sufficient
  information so that an invalidation message for a particular
  shared object $O$ is only sent to remote caches that
  actually hold a copy of $O$.
This is an unrealistic assumption because in an $N$-processor system,
  it requires roughly $N$ bits of state for each cached object!
However, it does ensure that at most one invalidation message
  is generated for each invalidation that actually occurs.
In that case, amortized RMR complexity corresponds well
  to amortized message complexity.

In realistic large-scale cache-coherent systems,
  the coherence protocol maintains far less state than in an ideal system,
  and so RMRs may trigger superfluous invalidation messages
  (i.e., ones that do not actually cause an invalidation).
In such systems, amortized RMR complexity may be lower asymptotically
  than amortized message complexity, and so
  our RMR complexity separation does not imply
  that in practice a large-scale cache-coherent machine could solve the \probname\
  more efficiently than a DSM machine with the same number of processors.

\section{Conclusion} \label{sec:conc}

In this paper, we showed that the \probname\ can be solved in a
   wait-free manner using only atomic reads and writes, with
   $O(1)$ RMRs per process in the CC model, and $O(1)$ space.
We then showed that the same problem cannot be solved in the DSM model
   with the same parameters, even if we weaken these parameters
   in all of the following ways simultaneously:
(1) we allow only one signaler instead of many,
(2) we settle for $O(1)$ RMRs complexity in the amortized sense rather than in the worst case,
(3) we allow unbounded space,
(4) we weaken the progress condition from wait-free to terminating, and
(5) we allow use of Compare-And-Swap or Load-Linked/Store-Conditional
    in addition to reads and writes.
Thus, we separate the CC and DSM models in terms
    of their power for solving the \probname\ efficiently with respect to
    amortized RMR complexity.


\paragraph{Acknowledgments}
Sincere thanks to Prof.\ Philipp Woelfel at the University of Calgary
for stimulating discussions on the subject of RMRs and the relative
power of the CC and DSM models for solving synchronization problems
efficiently with respect to amortized RMR complexity.
We are grateful to the anonymous referees for their insightful comments and suggestions.
Special thanks to Dr.\ Mark Tuttle at Intel for his perspective on the performance
  of CC systems in the real world and the shortcomings of the RMR
  complexity measure.
Thanks also to Dr.\ Ram Swaminathan at HP Labs and Dr.\ Robert Danek at Oanda
  for their careful proofreading of this work.

\bibliographystyle{abbrv}


\end{document}